%% file: main.tex
\documentclass[a4paper,11pt]{llncs}
\pdfoutput=1
\usepackage{times}
\usepackage{alltt}

\usepackage{latexsym}
\usepackage[all]{xy}
\usepackage{graphicx}

\usepackage{ntheorem}
\usepackage{amsmath}
\usepackage{amssymb}

\usepackage{stmaryrd}
\usepackage{moreverb,epsfig,url}
\usepackage[curve]{xypic}
\usepackage{tikz}
\usetikzlibrary{trees}
\usetikzlibrary{arrows}
\usepackage{listings}
\input{rapinclude}

\author{Thomas Genet}
\title{Towards Static Analysis of Functional Programs
\\ using Tree Automata Completion\thanks{This paper has been published in the
  Workshop on Rewriting Logic and Applications\cite{Genet-WRLA14}.}}

\institute{INRIA/IRISA, Universit\'e de Rennes, France\\
{\tt genet@irisa.fr}}

\begin{document}

\maketitle

\begin{abstract}
This paper presents the first step of a wider research effort to 
apply tree automata completion to the static analysis of functional programs. 
Tree Automata Completion is a family of techniques for computing or
approximating the set of terms reachable by a rewriting relation. 
The completion algorithm we focus on is parameterized by a set $E$ of equations
controlling the precision of the approximation and influencing its termination.
For completion to be used as a static analysis, the first step is to guarantee
its termination. In this work, we thus give a sufficient condition on $E$ and
$\TF$ for completion algorithm to always terminate. 
In the particular setting of functional programs, this condition can be 
relaxed into a condition on $E$ and $\TC$ (terms built on the set of
constructors) that is closer to what is done in the field of static analysis,
where abstractions are performed on data.
\end{abstract}

\section{Introduction}
Computing or approximating the set of terms reachable by rewriting
has more and more applications. For a Term Rewriting System (TRS) $\R$ and a
set of terms $L_0\subseteq \TF$, the set of reachable terms is $\desc(L_0)=\{t \in
\TF \sep \exists s\in L_0, s \rwR^* t\}$. This set can be computed exactly for
specific classes of $\R$~\cite{Genet-Habil} but, in general, it has to be
approximated. Applications of the approximation of $\desc(L_0)$ are ranging from
cryptographic protocol verification~\cite{avispa}, to static
analysis of various programming languages~\cite{BoichutGJL-RTA07}
or to TRS termination proofs~\cite{GeserHWZ-RTA05}. Most of
the techniques compute such approximations using tree automata as the core
formalism to represent or approximate the (possibly) infinite set of terms
$\desc(L_0)$. Most of them also rely on a Knuth-Bendix completion-like algorithm
completing a tree automaton $\A$ recognizing $L_0$ into an automaton $\A^*$
recognizing exactly, or over-approximating, the set $\desc(L_0)$. As a result,
these techniques can be refered as {\em tree automata completion}
techniques~\cite{Genet-RTA98,Takai-RTA04,FeuilladeGVTT-JAR04,BoichutCHK-IJFCS09,GenetR-JSC10,Lisitsa-RTA12}.
A strength of this algorithm, and at the same time a weakness, is that its
precision is parameterized by a function~\cite{FeuilladeGVTT-JAR04} or a set of
equations~\cite{GenetR-JSC10}. It is a strength because tuning the approximation
function (or equations) permits to adapt the precision of completion to a
specific goal to tackle. This is what made it successful for program and
protocol verification. On the other hand, this is a weakness because it is
difficult to guarantee its termination.

\noindent
In this paper, we define a simple
sufficient condition on the set of equations for the 
tree automata completion algorithm to terminate. This condition, which is strong
in general, reveals to be natural and well adapted for the approximation of
reachable terms when TRSs encode typed functional programs. We thus obtain a way to automatically
over-approximate the set of all reachable program states of a functional 
program, or even restrict it to the set of all results. Thus we can
over-approximate the image of a functional program.

\section{Related work}

{\em Tree automata completion}. 
With regards to most papers about
completion~\cite{Genet-RTA98,Takai-RTA04,FeuilladeGVTT-JAR04,BoichutCHK-IJFCS09,GenetR-JSC10,Lisitsa-RTA12},
our contribution is to give the first criterion {\em on the approximation} for the
completion to terminate. Note that it is possible to guarantee termination of
the completion by inferring an approximation adapted to the TRS under
concern, like in~\cite{OehlCKS-FASE03}. In this case, given a TRS, the approximation is
fixed and unique. Our solution is more flexible because it lets the user change
the precision of the approximation while keeping the termination
guarantee. In~\cite{Takai-RTA04}, T.~Takai have a completion parameterized by a
set of equations. He also gives a termination proof for its completion but only
for some restricted classes of TRSs. Here our termination proof holds for any 
left-linear TRS provided that the set of equations satisfy some properties.

\medskip
\noindent {\em Static analysis of functional programs}.  With regards to static
analysis of functional programs using grammars or automata, our contribution is
in the scope of data-flow analysis techniques, rather than control-flow
analysis. More precisely, we are interested here in predicting the results of a
function~\cite{OngR-POPL11}, rather than predicting the control
flow~\cite{Kobayashi-JACM13}. Those two papers, as well 
as many other ones, deal with higher order functions using complex higher-order
grammar formalisms (PMRS 
and HORS). Higher-order functions are not in the scope of the solution we
propose here. However, we obtained some preliminary results suggesting that an
extension to higher order functions is possible and gives relevant results (see
Section~\ref{sec:further}). Furthermore, using equations, approximations are
defined in a more declarative and flexible way than in~\cite{OngR-POPL11}, where
they are defined by a dedicated algorithm. Besides, the verification mechanisms
of~\cite{OngR-POPL11} use automatic abstraction refinement. This can be also
performed in the completion setting~\cite{BoichutBGL-ICFEM12} and adapted to
the analysis of functional programs~\cite{GenetS-rep13}. Finally, using a simpler
(first order) formalism, {\em i.e.} tree automata, makes it easier to take into
account some other aspects like: evaluation strategies and built-ins types (see
Section~\ref{sec:further}) that are not considered by those papers.

\section{Background}
\label{sec:bg}

In this section, we introduce some definitions and concepts that will
be used throughout the rest of the paper (see
also~\cite{BaaderN-book98,TATA}). Let $\F$ be a
finite set of symbols, each associated with an arity function, and let
$\X$ be a countable set of {\em variables}. $\TFX$ denotes the set of
{\em terms} and $\TF$ denotes the set of {\em ground terms} (terms
without variables). The set of variables of a term $t$ is denoted by
$\var(t)$. A {\em substitution} is a function $\sigma$ from $\X$ into
$\TFX$, which can be uniquely extended to an endomorphism of $\TFX$. A
{\em position} $p$ for a term $t$ is a finite word over $\NN$. The empty
sequence $\lambda$ denotes the top-most position. The set $\pos(t)$ of
positions of a term $t$ is inductively defined by $\pos(t)= \{
\lambda\} $ if $t \in \X$ or $t$ is a constant and $\pos(f(t_1,\dots,t_n)) = \{ \lambda \}
\cup \{i.p \mid 1 \leq i \leq n \et p \in \pos(t_i) \}$ otherwise.  If
$p \in \pos(t)$, then $t|_p$ denotes the subterm of $t$ at position
$p$ and $t[s]_p$ denotes the term obtained by replacement of the
subterm $t|_p$ at position $p$ by the term $s$.

A {\em term rewriting system} (TRS) $\R$ is a set of {\em rewrite
  rules} $l \rw r$, where $l, r \in \TFX$, $l \not \in \X$, and
$\var(l) \supseteq \var(r)$.  A rewrite rule $l \rw r$ is {\em
  left-linear} if each variable of $l$ occurs only once in $l$.  A TRS $\R$ is left-linear if
every rewrite rule $l \rw r$ of $\R$ is left-linear.  The TRS
$\R$ induces a rewriting relation $\rw_{\R}$ on terms as follows. Let
$s, t\in \TFX$ and $l \rw r \in \R$, $s \rw_{\R} t$ denotes that there
exists a position $p\in\pos(s)$ and a substitution $\sigma$ such that
$s|_p= l\sigma$ and $t=s[r\sigma]_p$.
Given a TRS $\R$, $\F$ can be split into two disjoint sets $\C$ and $\D$. All symbols
occurring at the root position of left-hand sides of rules of $\R$ are in
$\D$. $\D$ is the set of defined symbols of $\R$, $\C$ is the set of
constructors. Terms in $\TC$ are called {\em data-terms}.  
The reflexive transitive closure of $\rw_{\R}$ is denoted by
$\rw^*_{\R}$ and $s \rw^!_\R t$ denotes that $s\rw^*_\R t$ and $t$ is
irreducible by $\R$. The set of irreducible terms w.r.t. a TRS $\R$ is denoted
by $\lirr(\R)$.
The set of $\R$-descendants of a set of ground
terms $I$ is $\desc(I) = \{t \in \TF \sep \exists s \in I \sth s
\rw^*_{\R} t \}$.  
A TRS $\R$ is sufficiently complete if for all
$s\in\TF$,
$(R^*(\{s\})\cap\TC)\neq\emptyset$. 

An {\em equation set} $E$ is a set of {\em
  equations} $l = r$, where $l, r \in \TFX$.  
The relation $=_E$ is the smallest congruence such that for all
substitution $\sigma$ we have $l\sigma =_E r\sigma$. Given a TRS $\R$
and a set of equations $E$, a term $s\in\TF$ is rewritten modulo $E$
into $t\in\TF$, denoted $s \rwRE t$, if there exist $s'\in\TF$ and
$t'\in\TF$ such that $s=_Es'\rw_\R t'=_E t$. The reflexive transitive closure
$\rwRE^*$ of $\rwRE$ is defined as usual except that reflexivity is extended to
terms equal modulo $E$, {\em i.e.} for all $s, t\in\TF$ if $s=_E t$ then $s \rwRE^* t$.
The set of
$\R$-descendants modulo $E$ of a set of ground terms $I$ is $\edesc(I)
= \{t \in \TF \sep \exists s \in I \sth s \rwRE^* t \}.$

Let $\Q$ be a countably infinite set of symbols with arity $0$, called {\em states}, such
that $\Q \cap \F= \emptyset$.  $\TFQ$ is called the set of {\em
  configurations}. A {\em transition} is a rewrite rule $c \rw q$, where $c$ is
a configuration and $q$ is state.  A transition is {\em normalized} when $c =
f(q_1, \ldots, q_n)$, $f \in \F$ is of arity $n$, and $q_1, \ldots, q_n \in
\Q$. An {\em $\varepsilon$-transition} is a transition of the form $q \rw q'$
where $q$ and $q'$ are states. 
A bottom-up non-deterministic finite tree automaton ({\em tree automaton} for short)
over the alphabet $\F$ is a tuple $\A= \langle \F, \Q, \Q_F,\Delta \rangle$,
where $\Q_F$ is a finite subset of $\Q$, $\Delta$ is a finite set of normalized transitions and
$\varepsilon$-transitions. The transitive and reflexive {\em rewriting relation}
on $\TFQ$ induced by the set of transitions $\Delta$ (resp. all transitions
except $\varepsilon$-transitions) is denoted by $\rw_{\Delta}^*$
(resp. $\rwDnestar$). When $\Delta$ is attached to a tree automaton $\A$ we also
note those two relations $\rwA^*$ and $\rwAnestar$, respectively. A tree
automaton $\A$ is complete if for all $s\in\TF$ there exists a state $q$ of $\A$
such that $s \rwA^* q$. The language (resp. $\epsifree$-language) recognized by
$\A$ in a state $q$ is $\Lang(\A,q) = \{t \in \TF \sep t \rw^*_\A q \}$
(resp. $\Lang^{\epsifree}(\A,q) = \{ t \in \TF \sep t \rwAnestar q\}$). A state
$q$ of an automaton $\A$ is {\em reachable} (resp. $\epsifree$-reachable) if
$\Lang(\A,q)\neq \emptyset$ (resp. $\Lange(\A,q)\neq \emptyset$). 
We define $\Lang(\A)
= \bigcup_{q \in \Q_F} \Lang{}(\A, q)$.  
A set of transitions $\Delta$ is
$\epsifree$-deterministic if there are no two normalized transitions in $\Delta$
with the same left-hand side.  A tree automaton $\A$ is
$\epsifree$-deterministic if its set of transitions is $\epsifree$-deterministic.
Note that if $\A$ is $\epsifree$-deterministic then for all states $q_1, q_2$ of
$\A$ such that $q_1\neq q_2$, we have $\Lange(\A,q_1)\cap
\Lange(\A,q_2)=\emptyset$.

\section{Tree Automata Completion Algorithm}
\label{chap:theo}

Tree Automata Completion algorithms were proposed
in~\cite{Jacquemard-RTA96,Genet-RTA98,Takai-RTA04,GenetR-JSC10}. They
are very similar to a Knuth-Bendix completion except that they run on two distinct
sets of rules: a TRS $\R$ and a set of transitions $\Delta$ of a tree automaton
$\A$. 


\newcommand{\oldcomp}{{\cal C}_{\R,\alpha}}
\renewcommand{\nr}{\alpha}
\renewcommand{\comp}{{\cal C}_{\R,\alpha}}
\renewcommand{\desc}[1]{#1^{*}}

Starting from a tree automaton $\A_0=\langle \F, \Q, \Q_f,
\Delta_0\rangle$ and a left-linear
TRS $\R$, the algorithm computes a tree automaton $\A'$
such that $\Lang(\A') =\desc{\R}(\Lang(\A_0))$ or $\Lang(\A') \supseteq
\desc{\R}(\Lang(\A_0))$. 
The algorithm iteratively computes tree automata $\aaex^1$,
$\aaex^2$, \ldots such that $\forall i\geq 0: \Lang(\aaex^i) \subseteq
\Lang(\aaex^{i+1})$ until we get an automaton $\aaex^k$ with $k \in \NN$
and $\Lang(\aaex^k) = \Lang(\aaex^{k+1})$. For all $i\in\NN$, if $s \in
\Lang(\aaex^i)$ and $s \rwR t$, then $t\in \Lang(\aaex^{i+1})$. Thus, if
$\aaex^k$ is a fixpoint then it also verifies $\Lang(\aaex^k) \supseteq \desc{\R}(\Lang(\A_0))$.
To construct $\aaex^{i+1}$ from $\aaex^i$, we achieve a {\em completion step}
which consists in finding {\em critical pairs} between $\rw_{\R}$ and
$\rw_{\aaex^i}$. A critical pair is a triple $(l\rw r, \sigma, q)$ where $l \rw r
\in \R$, $\sigma:\X \mapsto \Q$ and $q\in\Q$ such that $l \sigma
\rw^{*}_{\aaex^{i}} q$ and $r \sigma \not \rw^{*}_{\aaex^{i}} q$.  For $r\sigma$ to be recognized by the same state and
thus model the rewriting of $l\sigma$ into $r\sigma$, it is enough to add the
necessary transitions to $\aaex^i$ to obtain $\aaex^{i+1}$ such that $r\sigma
\rw^*_{\aaex^{i+1}} q$. In~\cite{Takai-RTA04,GenetR-JSC10}, critical pairs are
joined in the following way: 
$$
\xymatrix{
 l\sigma \ar[r]_{\R}\ar[d]_{\aaex^i} & r\sigma \ar[d]^{\aaex^{i+1}}\\
 q & \ar[l]^{\aaex^{i+1}} q'
}
$$

\noindent
From an algorithmic point of view, there remains two problems to solve: find
all the critical pairs $(l\rw r, \sigma, q)$ and find the transitions
to add to $\aaex^i$ to have $r\sigma \rw_{\aaex^{i+1}}^* q$. The first
problem, called matching, can be efficiently solved using a specific
algorithm~\cite{FeuilladeGVTT-JAR04,Genet-Habil}. The second problem is solved
using Normalization.

\label{sec-tac}
\renewcommand{\oldcomp}{{\cal C}_{\R,\alpha}}
\renewcommand{\nr}{\alpha}
\renewcommand{\comp}{{\cal C}_{\R,\alpha}}
\renewcommand{\desc}[1]{#1^{*}}
\renewcommand{\desc}{\R^*}
\renewcommand{\nr}{E}
\renewcommand{\comp}{{\cal C}}

\subsection{Normalization}
\label{sec-norm}
The normalization function replaces subterms either by states of $\Q$
(using transitions of $\Delta$) or by new states. A state $q$ of $\Q$ is used to
normalize a term $t$ if $t \rwDne q$. Normalizing by reusing states of
$\Q$ and transitions of $\Delta$ permits to preserve the $\epsifree$-determinism of
$\rwDne$. Indeed, $\rwDne$ can be kept deterministic during completion though
$\rwD$ cannot.

\begin{definition}[New state]
  Given a set of transitions $\Delta$, a new state (for $\Delta$) is a state 
of $\Q\setminus\Q_f$ not occurring in left or right-hand sides of rules of
  $\Delta$~\footnote{Since $\Q$ is a countably infinite set of 
states, $\Q_f$ and $\Delta$ are finite, a new state can always be found.}.
\end{definition}

\noindent
We here define normalization as a bottom-up process. This definition is simpler and
equivalent to top-down definitions~\cite{GenetR-JSC10}. In the recursive call,
the choice of the context $C[\,]$ may be non deterministic but all the
possible results are the equivalent modulo state renaming. 

\begin{definition}[Normalization]
\label{def:normalization}
Let $\Delta$ be a set of transitions defined on a set of states $\Q$, $t \in
\TFQ \setminus \Q$. Let $C[\;]$ be a non empty context of
$\TFQ\setminus \Q$, $f\in \F$ of arity 
$n$, and $q,q',q_1, \ldots, q_n\in\Q$. The normalization function is inductively
defined by:
\begin{enumerate}
\item $\Norm_{\Delta}(f(q_1, \ldots, q_n) \rw q) = \{ f(q_1, \ldots, q_n) \rw
  q\}$

\item $\Norm_{\Delta}(C[f(q_1, \ldots, q_n)] \rw q) =$ 
\begin{tabular}[t]{l}
    $ \{ f(q_1, \ldots, q_n) \rw q'\} \: \cup$ \\
    $\Norm_{\Delta\cup\{f(q_1, \ldots, q_n) \rw  q'\}} (C[q'] \rw q)$
\end{tabular} 

where either ($f(q_1, \ldots, q_n) \rw q' \in \Delta$) or ($q'$ is a new state for
$\Delta$ and
$\forall q'' \in Q : f(q_1, \ldots, q_n) \rw q'' \not \in \Delta$).
\end{enumerate} 
\end{definition}

\noindent
In the second case of the definition, if there are several states $q'$ such that $f(q_1,\ldots,q_n) \rw
q' \in \Delta$, we arbitrarily choose one of them. We illustrate the above definition on the
normalization of a simple transition.
\begin{example}
Given $\Delta=\{b \rw q_0\}$, 
$\Norm_\Delta(f(g(a), b, g(a))\rw q) = \{a \rw q_1, g(q_1) \rw q_2, b \rw q_0, f(q_2, q_0, q_2) \rw q\}$
\end{example}

\subsection{One step of completion}
\label{sec:eqcompletionalgo}

A step of completion only consists in joining critical pairs. We first need to
formally define the substitutions under concern: {\em state substitutions}.

\begin{definition}[State substitutions, $\Sigma(\Q,\X)$]
\label{def:qsubst}
A {\em state substitution} over an automaton
$\A$ with a
set of states $\Q$ is a function $\sigma: \X \mapsto \Q$. We can extend this definition to a morphism
$\sigma: \TFX \mapsto \mathcal{T}(\mathcal{F},\mathcal{Q})$.
We denote by $\Sigma(\Q, \X)$ the set of state substitutions built over $\Q$ and $\X$.
\end{definition}

\begin{definition}[Set of 
critical pairs]
Let a TRS $\R$ and a tree automaton $\A=\aut$. The set of 
critical
pairs between $\R$ and $\A$ is $CP(\R,\A)=\{(l\rw r, \sigma, q) \sep l
\rw r \in \R, \: q \in \Q, \: \sigma\in \Sigma(\Q, \X), \: l\sigma \rw_{\A}^* q,
\: r\sigma \not\rw_{\A}^*q 
\}$. 
\end{definition}

\noindent
Recall that the completion process builds a sequence
$\aaex^0,\aaex^1,\ldots,\aaex^k$ of automata such that if $s\in\Lang(\aaex^i)$
and $s\rw_{\R} t$ then $t\in\Lang(\aaex^{i+1})$. One step of completion,
{\em i.e.} the process computing $\aaex^{i+1}$ from $\aaex^i$, is defined as follows.
Again, the following definition is a simplification of the definition of~\cite{GenetR-JSC10}.

\begin{definition}[One step of completion]
\label{def:completion-one-step}
  Let $\A= \langle \F, \Q, \Q_f, \Delta \rangle$ be a tree automaton, 
  $\R$ be a left-linear TRS. 
  The one step completed automaton is $\comp_{\R}(\A)= \langle \F, \Q,$ $\Q_f,
  Join^{CP(\R,\A)}(\Delta) \rangle$ where $Join^{S}(\Delta)$ is inductively defined by:
\begin{itemize}
\item $Join^{\emptyset}(\Delta)= \Delta$
\item $Join^{\{(l\rw r, q, \sigma)\} \cup S}(\Delta)=  Join^S(\Delta \cup
  \Delta')$ where 


$\Delta'= \{q' \rw q\}$ if there exists $q'\in \Q$ s.t. $r\sigma \rwDnestar q'$, and
otherwise

$\Delta'= \Norm_{\Delta}(r\sigma \rw q') \cup \{q' \rw q\}$ where $q'$ is a new state for
  $\Delta$

\end{itemize}
\end{definition}

\begin{example}
Let $\A$ be a tree automaton with $\Delta=\{f(q_1) \rw q_0, a \rw q_1, g(q_1)
\rw q_2\}$. If $\R=\{f(x) \rw f(g(x))\}$ then $CP(\R,\A)= \{(f(x) \rw f(g(x)),
  \sigma_3, q_0)\}$ with $\sigma_3= \{x \mapsto q_1\}$, because $f(x) \sigma_3 \rwA^*
  q_0$ and $f(x) \sigma_3 \rwR f(g(x))\sigma_3$. We have $f(g(x))\sigma_3=
  f(g(q_1))$ and there exists no
  state $q$ such that $f(g(q_1)) \rwAnestar q$. Hence,
  $Join^{\{(f(x) \rw f(g(x)),\sigma_3, q_0)\}}(\Delta)= Join^{\emptyset}(\Delta
  \cup \Norm_{\Delta}(f(g(q_1)) \rw q_3) \cup \{q_3 \rw q_0\})$. Since
  $\Norm_{\Delta}(f(g(q_1)) \rw q_3)= \{f(q_2) \rw q_3, q(q_1) \rw q_2\}$, we
  get that $\comp_{\R}(\A)= \langle \F, \Q \cup \{q_3\}, \Q_f,\Delta \cup
  \{f(q_2) \rw q_3, q_3 \rw q_0\}\rangle$.
\end{example}

\subsection{Simplification of  Tree Automata by Equations}
\label{sec:simplif}
In this section, we define the {\em simplification} of tree automata $\A$
w.r.t. a set of equations $E$. This operation permits to
over-approximate languages that cannot be recognized {\em exactly} using tree
automata completion, {\em e.g.} non regular languages. The simplification operation
consists in finding $E$-equivalent terms recognized in $\A$ by different states
and then by merging those states together. The merging of states is performed
using renaming of a state in a tree automaton.

\label{sec:merging}
\begin{definition}[Renaming of a state in a tree automaton]
  Let $\Q,\Q'$ be set of states, $\A= \langle \F, \Q, \Q_f, \Delta \rangle$ be a
  tree automaton, and $\alpha$ a function $\alpha : \Q \mapsto \Q'$. We denote by
  $\A\alpha$ the tree automaton where every occurrence of $q$ is
  replaced by $\alpha(q)$ in $\Q$, $\Q_f$ and in every left and right-hand
  side of every transition of $\Delta$. 
\end{definition}
If there exists a bijection $\alpha$ such that $\A=\A'\alpha$ then $\A$ and
$\A'$ are said to be {\em equivalent modulo renaming}. 
Now we define the {\em simplification relation} which merges states in a tree
automaton according to an equation. Note that it is not required for equations
of $E$ to be linear.

\begin{definition}[Simplification relation]
\label{def:simprel}
  Let $\A= \langle \F, \Q, \Q_f, \Delta \rangle$ be a tree automaton and $E$ be
  a set of equations. For $s=t \in E$, $\sigma \in \Sigma(\Q, \X)$, $q_a,
  q_b \in \Q$ such that $s\sigma\rwAnestar q_a$, $t\sigma\rwAnestar q_b$, 
 and $q_a \neq q_b$ then $\A$ can be {\em simplified} into $\A'= \A\{q_b \mapsto
  q_a\}$, denoted by
$\A \simp_E \A'$.
\end{definition}

\begin{example}
  Let $E=\{s(s(x))=s(x)\}$ and $\A$ be the tree automaton with set of
  transitions $\Delta= \{ a \rw q_0, s(q_0) \rw q_1, s(q_1) 
  \rw q_2\}$. We can perform a simplification step using the equation
  $s(s(x))=s(x)$ because we found a substitution $\sigma=\{x
  \mapsto q_0\}$ such that: $s(s(x))\sigma \rwAnestar q_2$ and $s(x)\sigma\rwAnestar q_1$
  Hence, $\A \simp_E \A'= \A \{q_2 \mapsto q_1\}$\footnote{or $\{q_1 \mapsto q_2\}$,
any of $q_1$ or $q_2$ can be used for renaming.}
\end{example}

\noindent
As stated in~\cite{GenetR-JSC10}, simplification
$\simp_E$ is a terminating relation (each step suppresses a state) and it
enlarges the language recognized by a tree automaton, {\em i.e.} if $\A \simp_E \A'$
then $\Lang(\A) \subseteq \Lang(\A')$. Furthermore, no matter how
simplification steps are performed, the obtained automata are equivalent modulo
state renaming. In the following, $\A \simp_E^! \A'$ denotes that $\A \simp_E^*
\A'$ and $\A'$ is irreducible by $\simp_E$. We denote by $\simpleq E\A$ any automaton $\A'$ such that $\A
\simp_E^! \A'$. 


\begin{theorem}[Simplified Tree Automata~\cite{GenetR-JSC10}]
  \label{theorem:canonical}
  Let $\A,\A_1', \A'_2$ be tree automata and $E$ be a set of 
  equations. If $\A \simp_E^! \A_1'$ and $\A \simp_E^! \A_2'$ then
  $\A'_1$ and $\A'_2$ are equivalent modulo state renaming.
\end{theorem}

\subsection{The full Completion Algorithm}
\label{sec:completion}

\begin{definition}[Automaton completion]
\label{def:completion}
Let $\A$ be a tree automaton, $\R$ a left-linear TRS and $E$ a set of equations. 
\begin{itemize}
\item $\aapprox^0= \A$
\item $\aapprox^{n+1}= \simpleq E {\comp_{\R}(\aapprox^n)}$, for $n \geq 0$
\end{itemize}
If there exists $k \in \NN$ such that $\aapprox^{k}=\aapprox^{k+1}$, then we
denote $\aapprox^k$ by $\aapprox^*$.
\end{definition}
In practice, checking if $CP(\R,\aapprox^k)=\emptyset$ is sufficient to know
that $\aapprox^k$ is a fixpoint. However, a fixpoint cannot always be finitely
reached\footnote{See~\cite{Genet-Habil}, for classes of $\R$ for which a
  fixpoint always exists.}. To ensure termination, one can provide a set of
approximating equations to overcome infinite rewriting and 
completion divergence. 

\begin{example}
  Let $\R=\{f(x, y) \rw f(s(x), s(y))\}$, $E=\{ s(s(x))=s(x) \}$ and $\A^0$ be
  the tree automaton with set of transitions $\Delta=\{f(q_a, q_b) \rw q_0), a
  \rw q_a, b \rw q_b\}$, {\em i.e.} $\Lang(\A^0)=\{f(a,b)\}$. The completion ends
  after two completion steps on $\aapprox^2$ which is a fixpoint. Completion
  steps are summed up in the following table. To simplify the presentation, we do
  not repeat the common transitions: $\aapprox^i$ and $\comp_\R(\A^{i})$ columns
  are supposed to contain all transitions of $\A^0,\ldots, \aapprox^{i-1}$.
The automaton $\aapprox^1$ is exactly $\comp_{\R}(\A^0)$ since
simplification by equations do not apply. Simplification has been applied on
$\comp_\R(\aapprox^1)$ to obtain $\aapprox^2$.

\noindent
{\small
\begin{center}
\begin{tabular}{||r||r|r||r|r||}
\hline 
$\A^0$ & $\comp_\R(\A^0)$ & $\aapprox^1$ & $\comp_\R(\aapprox^1)$ & $\aapprox^2$\\
\hline
$f(q_a, q_b) \rw q_0$ & $f(q_1, q_2) \rw q_3$ & $f(q_1, q_2) \rw q_3$ & $f(q_4,q_5) \rw q_6$ & $f(q_1, q_2) \rw q_6$\\
$a \rw q_a$ & $s(q_a) \rw q_1$ & $s(q_a) \rw q_1$ & $s(q_1)\rw q_4$ & $s(q_1) \rw q_1$ \\ 
$b \rw q_b$ & $s(q_b) \rw q_2$ & $s(q_b) \rw q_2$ & $s(q_2) \rw q_5$ & $s(q_2) \rw q_2$ \\
            & $q_3 \rw q_0 $ & $q_3 \rw q_0 $ & $q_6 \rw q_3 $ & \\
\hline
\end{tabular}
\end{center}}
\end{example}

\noindent
Now, we recall the lower and upper bound theorems.  Tree automata completion of
automaton $\A$ with TRS $\R$ and set of equations $E$ is lower bounded by
$\R^*(\Lang(\A))$ and upper bounded by $\R^*_{E}(\Lang(\A))$. The lower bound theorem 
ensures that the completed automaton $\aapprox^*$ recognizes all $\R$-reachable
terms (but not all $\R/E$-reachable terms). The upper bound theorem guarantees
that all terms recognized by $\aapprox^*$ are only $\R/E$-reachable terms. 

\begin{theorem}[Lower bound~\cite{GenetR-JSC10}]
  \label{completeness}
  Let $\R$ be a left-linear TRS, $\A$ be a tree automaton and $\nr$ be a set of
  equations. If completion terminates on $\aapprox^*$ then
$\Lang(\aapprox^*)\supseteq \desc(\Lang(\A))$.
\end{theorem}

\noindent
The upper bound theorem states the precision result of
completion. It is defined using the $\R/E$-coherence property.
The intuition behind $\R/E$-coherence is the following: in the tree automaton
$\varepsilon$-transitions represent rewriting steps and normalized transitions
recognize $E$-equivalence classes. More precisely, in a $\R/E$-coherent tree
automaton, if two terms $s,t$ are recognized into the same state $q$  using
only normalized transitions then they belong to 
the same $E$-equivalence class. Otherwise, if at least one
$\varepsilon$-transition is necessary to recognize, say, $t$ into $q$ then at
least one step of rewriting was necessary to obtain $t$ from $s$.

\begin{theorem}[Upper bound~\cite{GenetR-JSC10}]
  \label{correctness}
  Let $\R$ be a left-linear TRS, $E$ a set of equations and $\A$ a
  $\R/E$-coherent tree automaton. For any $i\in\NN$:
%
%
$\Lang(\aapprox^i)\subseteq \desc_E(\Lang(\A))$ 
and $\aapprox^i$ is $\R/E$-coherent. 
\end{theorem}


\section{Termination criterion for a given set of equations}
\label{sec:termCrit}
\renewcommand{\desc}{\R^*}

Given a set of
equations $E$, the effect of the simplification with $E$ on a tree automaton is
to merge two distinct states recognizing instances of the left and right-hand
side for all the equations of $E$. In this section, we give a sufficient
condition on $E$ and on the completed tree automata $\aapprox^i$ for the tree
automata completion to always terminate. The intuition behind this condition is
simple: if the set of equivalence classes for $E$, {\em i.e.} $\TF/_{=_E}$,  is finite then so should be
the set of new states used in completion. However, this is not true in general
because simplification of an automaton with $E$ does not necessarily merge all
$E$-equivalent terms.

\begin{example}
\label{ex:compNotEcompat}
  Let $\A$ be the tree automaton with set of transitions $a \rw q$, $\R=\{a \rw
  c\}$ and let $E=\{a=b,b=c\}$. The set of transitions of $\comp_\R(\A)$ is $\{a
  \rw q, c \rw q', q' \rw q\}$. We have $a =_E c$, $a \in
  \Lange(\comp_\R(\A),q)$ and $c \in \Lange(\comp_\R(\A),q')$ but on the
  automaton $\comp_\R(\A)$, no simplification situation (as described by
  Definition~\ref{def:simprel}), can be found because the term $b$ is not
  recognized by $\comp_\R(\A)$. Hence, the simplified automaton is
  $\comp_\R(\A)$ where $a$ and $c$ are recognized by different states.
\end{example}

\noindent
There is no simple solution to have a simplification algorithm
merging all states recognizing $E$-equivalent terms (see
Section~\ref{sec:further}). 
Having a complete automaton $\A$ solve the above problem but leads to rough
approximations (see~\cite{Genet-rep13}).
In the next section, we propose to give some simple restrictions on $E$ to
ensure that completion terminates. In Section~\ref{sec:functional}, we will see 
how those restrictions can easily be met for ``functional'' TRS, {\em i.e.} a
typed first-order functional program translated into a TRS. 

\subsection{General criterion}
What Example~\ref{ex:compNotEcompat} shows is that, for a simplification with
$E$ to apply, it is necessary that both sides of the equation are recognized by
the tree automaton. In the following, we will define a set $\Ec$ of {\em
  contracting} equations so that this property is true. What
Example~\ref{ex:compNotEcompat} does not show is that, by default, tree automata
are not $E$-compatible. In particular, any non $\epsifree$-deterministic automaton does not
satisfy the reflexivity of $=_E$. For instance, if an automaton $\A$ has two
transitions $a \rw q_1$ and $a \rw q_2$, since $a=_E a$ for all $E$, for $\A$ to
be $E$-compatible we should have $q_1 = q_2$.  To enforce
$\epsifree$-determinism by automata simplification, we define a set of {\em reflexivity equations} as follows.

\begin{definition}[Set of reflexivity equations $\ETF$]
For a given set of symbols $\F$, $\ETF=\{f(x_1, \ldots, x_n)=f(x_1, \ldots, x_n) \sep
f \in \F, \mbox{ and arity of $f$ is $n$}\}$, where $x_1\ldots x_n$ are pairwise
distinct variables.
\end{definition}

\noindent
Note that for all set of equations $E$, the relation $=_E$ is trivially equivalent to
$=_{E\cup \ETF}$. Furthermore, simplification with $\ETF$ transforms 
all automaton into an $\epsifree$-deterministic automaton, as stated in the
following lemma.

\begin{lemma}
\label{lem:etf_det}
For all tree automaton $\A$ and all set of equation $E$, if $E\supseteq \ETF$
and $\A \simp_E^! \A'$ then $\A'$ is $\epsifree$-deterministic. 
\end{lemma}

\begin{proof} Shown by induction on the height of terms
 (see~\cite{Genet-rep13} for details). $\qed$
\end{proof}
\noindent
We now define sets of contracting equations. Such sets are defined for
a set of symbols $\myF$ which can be a subset of $\F$. This will be used later
to restrict contracting equations to the subset of constructor symbols of $\F$.

\begin{definition}[Sets of contracting equations for $\myF$, $\Ecpy$]
Let $\myF \subseteq \F$. A set of equations is contracting for $\myF$,
denoted by $\Ecpy$, if all equations of $\Ecpy$ are of the form $u= u|_p$ with $u\in \myTFX$ a linear term, $p \neq \lambda$, and 
if the set of normal forms of $\myTF$ w.r.t. the TRS $\ERcpy=\{ u \rw u|_p \: \sep
\: u= u|_p \in \Ecpy\}$ is finite.

\end{definition}

\noindent
Contracting equations, if defined on $\F$,  define an upper bound on the number
of states of a simplified automaton.

\begin{lemma}
\label{lem:cardComp}
Let $\A$ be a tree automaton and $\Ecf$ a set of contracting equations for
$\F$. If $E\supseteq \Ecf \cup \ETF$ then the simplified automaton $\simpleq E \A$
is an $\epsifree$-deterministic automaton having no more states than terms in
$\lirr(\ERcf)$.  
\end{lemma}
\begin{proof}
  First, assume for all state $q$ of $\simpleq E \A$, $\Lange(\simpleq E \A,q)
  \cap \lirr(\ERcf) =\emptyset$. Then, for all terms $s$ such that $s \rwne q$, we know
  that $s$ is not in normal form w.r.t. $\ERcf$. As a result, the left-hand side
  of an equation of $\Ecf$ can be applied to $s$. This means that there exists an
  equation $u = u|_p$, a ground context $C$ and a substitution $\theta$ such
  that $s=C[u\theta]$. Furthermore, since $s \rwne q$, we know
  that $C[u\theta]\rwne q$ and that there exists a state $q'$
  such that $C[q'] \rwne q$ and $u\theta \rwne q'$. From $u\theta\rwne q'$, we know that all subterms of
  $u\theta$ are recognized by at least one state in $\simpleq E \A$. Thus, there
  exists a state $q''$ such that $u|_p\theta \rwne q''$. We thus
  have a situation of application of the equation $u = u|_p$ in the
  automaton. Since $\simpleq E \A$ is simplified, we thus know that $q'=q''$. As
  mentioned above, we know that $C[q'] \rwne q$. Hence
  $C[u|_p\theta] \rwne C[q'] \rwne q$. If
  $C[u|_p\theta]$ is not in normal form w.r.t. $\ERcf$ then we can do the same reasoning
  on $C[u|_p\theta] \rwne q$ until getting a term that is in
  normal form w.r.t. $\ERcf$ and recognized by the same state $q$. Thus, this
  contradicts the fact that $\simpleq E \A$ recognizes no term of $\lirr(\ERcf)$.

Then, by definition of $\Ecf$, $\lirr(\ERcf)$ is finite. Let $\{t_1, \ldots,
t_n\}$ be the subset of $\lirr(\ERcf)$ recognized by $\simpleq E
\A$. Let $q_1, \ldots, q_n$ be the states recognizing $t_1, \ldots, t_n$
respectively. We know that there is a finite set of states recognizing $t_1,
\ldots, t_n$ because $E \supseteq
\ETF$ and Lemma~\ref{lem:etf_det} entails that $\simpleq E \A$ is
$\epsifree$-deterministic. Now, for all terms $s$ recognized by a state $q$ in $\simpleq E
\A$, {\em i.e.} $s \rwne q$,  we can use a reasoning similar to the
one carried out above and show that $q$ is equal to one state of $\{q_1, \ldots,
q_n\}$ recognizing normal forms of $\ERcf$ in $\simpleq E \A$. Finally, there
are at most $card(\lirr(\ERcf))$ states in $\simpleq E \A$. $\qed$
\end{proof}

\noindent
Now it is possible to state the Theorem guaranteeing the termination of
completion if the set of equations $E$ contains a set of contracting equations
$\Ecf$ for $\F$ and a set of reflexivity equations. 

\begin{theorem}
\label{th:termFinite}
Let $\A$ be a tree automaton, $\R$ a left linear TRS and $E$ a set of
equations. If $E\supseteq \ETF \cup \Ecf$, then completion of $\A$ by $\R$ and
$E$ terminates.
\end{theorem}

\begin{proof}
For completion to diverge it must produce infinitely many new states. This is
impossible if $E$ contains $\Ecf$ and $\ETF$ (see
Lemma~\ref{lem:cardComp}). $\qed$
\end{proof}

\subsection{Criterion for Functional TRSs}
\label{sec:functional}
Now, we consider functional programs viewed as TRSs. We assume that such
TRSs are left-linear, which is a common assumption on TRSs obtained from
functional programs~\cite{BaaderN-book98}. In this section, we will restrict
ourselves to sufficiently complete TRSs obtained from functional programs and
will refer to them as {\em functional TRSs}. For TRSs representing functional programs,  
defining contracting equations of $\Ecc$ on $\C$ rather than
on $\F$ is enough to guarantee termination of completion. This is more convenient and
also closer to what is usually done in static analysis where abstractions are
usually defined on data and not on function applications.
Since the TRSs we consider are sufficiently complete, any term of $\TF$ can be
rewritten into a data-term of $\TC$. As above, using equations of $\Ecc$ we are going to
ensure that the data-terms of the computed languages will be
recognized by a bounded set of states. To lift-up this property to $\TF$ it is
enough to ensure that $\forall s,t \in \TF$ if $s \to_R t$ then $s$
and $t$ are recognized by equivalent states. This is the role of the set of
equations $E_R$.

\begin{definition}[$E_\R$]
Let $\R$ be a TRS, the set of $\R$-equations is $E_\R= \{l=r \sep l \rw r \in \R\}$.
\end{definition}



\begin{theorem}
\label{th:termComplete}
Let $\A_0$ be a tree automaton, $\R$ a sufficiently complete left-linear TRS and $E$ a set of equations.
If $E \supseteq \ETF \cup \Ecc \cup E_\R$ with $\Ecc$ contracting then completion of $\A_0$ by $\R$ and $E$
terminates.
\end{theorem}

\begin{proof}
  Firstly, to show that the number of states recognizing terms
  of $\TC$ is finite we can do a proof similar to the one of
  Lemma~\ref{lem:cardComp} . Let $G\subseteq \TC$ be the finite set of normal forms
  of $\TC$ w.r.t. $\ERcc$. Since $E \supseteq \ETF \cup \Ecc$, like in the proof
  of Lemma~\ref{lem:cardComp}, we can show that in any completed automaton,
  terms of $\TC$ are recognized by no more states than terms in $G$. Secondly,
  since $\R$ is sufficiently complete, for all terms $s \in \TF \setminus \TC$ we
  know that there exists a term $t \in \TC$ such that $s \rwR^* t$. The
  fact that $E\supseteq E_\R$ guarantees that $s$ and $t$ will be recognized by
  equivalent states in the completed (and simplified) automaton. Since the
  number of states necessary to recognize $\TC$ is finite, so is the number of
  states necessary to recognize terms of $\TF$. $\qed$
\end{proof}

\noindent
Finally, to exploit the types of the functional program, we now see $\F$ as a
many-sorted signature whose set of sorts is $\sorts$. Each symbol $f\in
\F$ is associated to a profile $f : S_1 \times \ldots \times S_k \mapsto
S$ where $S_1, \ldots, S_k, S \in \sorts$ and $k$ is the arity of $f$. 
Well-sorted terms are inductively defined as follows: $f(t_1, \ldots, t_k)$ is a
well-sorted term of sort $S$ if $f : S_1 \times \ldots \times S_k \mapsto S$ and
$t_1, \ldots, t_k$ are well-sorted terms of sorts $S_1, \ldots, S_k$,
respectively. We denote by $\TFXS$, $\TFS$ and $\TCS$ the set of well-sorted
terms, ground terms and constructor terms, respectively. Note that we have
$\TFXS \subseteq \TFX$, $\TFS \subseteq \TF$ and $\TCS \subseteq \TC$.  
We assume that $\R$ and $E$ are {\em sort preserving}, {\em
  i.e.} that for all rule $l \rw r \in R$ and all equation $u=v \in E$, $l,r,u,v
\in \TFXS$, $l$ and $r$ have the same sort and so do $u$ and $v$. 
Note that well-typedness of the functional program entails the well-sortedness
of $\R$. We
still assume that the (sorted) TRS is sufficiently 
complete, which is defined in a similar way except that it holds only for
well-sorted terms, {\em i.e.}  for all $s\in \TFS$ there exists a term
$t\in\TCS$ such that $s \rwR^* t$.
We slightly refine the definition of contracting equations as follows. For all
sort $S$, if $S$ has a unique constant symbol we note it $c^S$.

\begin{definition}[Set $\Ecps$ of contracting equations for $\myF$ and $\sorts$]
Let $\myF \subseteq \F$. The set of well-sorted equations $\Ecps$ is {\em
  contracting} (for $\myF$) if its equations are of the form (a) $u= u|_p$ with $u$ linear and $p \neq \Lambda$, or
(b) $u=c^S$ with $u$ of sort $S$, and if the set of normal forms of $\myTFs$
w.r.t. the TRS $\ERcps=$ $\{ u \rw
v \sep u=v \in \Ecps \land (v=u|_p \lor v=c^S)\}$ is
finite.
\end{definition}

\noindent
The termination theorem for completion of sorted TRSs is similar to the
previous one except that
it needs $\R/E$-coherence of $\A_0$ to ensure that terms
recognized by completed automata are well-sorted (see~\cite{Genet-rep13} for proof).

\begin{theorem}
\label{th:termSort}
Let $\A_0$ be a tree automaton recognizing well-sorted terms, $\R$ a sufficiently
complete sort-preserving left-linear TRS and $E$ a sort-preserving set of
equations. If $E \supseteq \ETF \cup \Eccs \cup E_\R$ with $\Eccs$ contracting and $\A_0$ is
$\R/E$-coherent then completion of $\A_0$ by $\R$ and $E$  
terminates.
\end{theorem}

\subsection{Experiments}
\label{sec:experiments}
The objective of data-flow analysis is to predict the set of all program states
reachable from a language of initial function calls, {\em i.e. } to over-approximate
$\desc(\Lang(\A))$ where $\R$ represents the functional program and $\A$ the
language of initial function calls. In this setting, we automatically
compute an automaton $\aapprox^*$ over-approximating $\desc(\Lang(\A))$. But we
can do more. Since
we are dealing with left-linear TRS, it is possible to build $\airr{\R}$
recognizing $\lirr(\R)$. Finally, since tree automata are closed under all boolean
operations, we can compute an approximation of all the results of the function
calls by computing the tree automaton recognizing the intersection between
$\aapprox^*$ and $\airr{\R}$. 

Here is an example of application of those theorems.
Completions are performed using \timbuk.
All the $\airr{\R}$ automata and intersections were
performed using Taml. 
Details can be found in~\cite{GenetS-rep13}.

{\footnotesize
\begin{alltt}
\Ops append:2 rev:1 nil:0 cons:2 a:0 b:0   \Vars X Y Z U Xs    
\TRS R
append(nil,X)->X     append(cons(X,Y),Z)->cons(X,append(Y,Z))                         
rev(nil)->nil        rev(cons(X,Y))->append(rev(Y),cons(X,nil))

\Automaton A0 \States q0 qla qlb qnil qf qa qb \FinalStates q0 \Transitions
rev(qla)->q0         cons(qb,qnil)->qlb    cons(qa,qla)->qla    nil->qnil
cons(qa,qlb)->qla    a->qa                 cons(qb,qlb)->qlb    b->qb

\Equations E \Rules    cons(X,cons(Y,Z))=cons(Y,Z)  %%% Ec
%%% E_R                                     %%% E^r
append(nil,X)=X                             rev(X)=rev(X)   
append(cons(X,Y),Z)=cons(X,append(Y,Z))     cons(X,Y)=cons(X,Y)    
rev(nil)=nil                                append(X,Y)=append(X,Y)                     
rev(cons(X,Y))=append(rev(Y),cons(X,nil))   a=a  b=b  nil=nil                            
\end{alltt}
}

\noindent
In this example, the TRS $\R$ encodes the classical {\em reverse} and {\em
  append} functions. The language recognized by automaton $\A_0$ is the set of
terms of the form $rev([a, a, \ldots,$ $b, b, \ldots])$. Note that there are at
least one $a$ and one $b$ in the list. We assume that $\sorts=\{T,list\}$ and
sorts for symbols are the following: $a:T$, $b:T$, $nil: list$, $cons: T \times
list \mapsto list$, $append: list \times list \mapsto list$ and $rev: list
\mapsto list$. Now, to use Theorem~\ref{th:termSort}, we need to prove each of
its assumptions.  The set $E$ of equations contains $E_\R$, $\ETF$ and
$\Eccs$. The set of Equations $\Eccs$ is contracting because the automaton
$\airr{\ERccs}$ recognizes a finite language. This automaton can be computed
using Taml: it is the intersection between the automaton
$\A_{\TCS}$\footnote{Such an automaton has one
  state per sort and one transition per constructor. For instance, on our
  example $\A_{\TCS}$ will have transitions: $a \rw qT$, $b \rw qT$,
  $cons(qT,qlist) \rw qlist$ and $nil \rw qlist$.}  recognising $\TCS$ and the
automaton $\airr{\{cons(X,cons(Y,Z)) \rw cons(Y,Z)\}}$:

{\footnotesize
\begin{alltt}
\States q2 q1 q0 \FinalStates q0 q1 q2 
\Transitions b->q2 a->q2 nil->q1 cons(q2,q1)->q0
\end{alltt}
}

\noindent
The language of $\A_0$ is well-sorted and $E$ and $\R$ are sort preserving.  We can
prove sufficient completeness of $\R$ on $\TFS$ using, for instance,
Maude~\cite{maude-manual} or even \timbuk~\cite{Genet-RTA98} itself. 
The last
assumption of Theorem~\ref{th:termSort} to prove is that $\A_0$ is
$\R/E$-coherent. This can be shown by remarking that each state $q$ of $\A_0$
recognizes at least one term and if $s \rwAznestar q$ and $t
\rwAznestar q$ then $s \equiv_E t$. For instance $cons(b, cons(b,nil))
\rwAznestar q_{lb}$ and $cons(b, nil) \rwAznestar q_{lb}$ and $cons(b,
cons(b,nil))\equiv_E cons(b, nil)$. Thus, completion is guaranteed to terminate:
after 4 completion steps (7~ms) we obtain a fixpoint automaton $\aapprox^*$
with 11 transitions. To restrain the language to normal forms it is
enough to compute the intersection with $\lirr(R)$.
Since we are dealing with sufficiently complete TRSs, we know that
$\lirr(R) \subseteq \TCS$. Thus, we can use again $\A_{\TCS}$ for the
intersection that is:
 
{\footnotesize
\begin{alltt}
\States q3 q2 q1 q0  \FinalStates q3  \Transitions a->q0  nil->q1  b->q2  
cons(q0,q1)->q3  cons(q0,q3)->q3  cons(q2,q1)->q3  cons(q2,q3)->q3
\end{alltt}
}
 
\noindent
which recognizes any (non empty) flat list of $a$ and $b$. Thus, our analysis  preserved
the property that the result cannot be the empty list but lost the order of the
elements in the list. This is not surprising because
the equation {\tt \footnotesize cons(X, cons(Y, Z))=cons(X, Z)}
makes $cons(a, cons(b, nil))$ equal to $cons(a,nil)$. It is possible to refine
by hand $\Eccs$ using the following equations: {\tt \footnotesize
  cons(a,cons(a,X))=cons(a,X)}, {\tt \footnotesize cons(b,cons(b,X))=cons(b,X)},
  {\tt \footnotesize cons(a,cons(b,cons(a,X)))=cons(a,X)}.
This set of equations avoids the previous problem. Again, 
$E$ verifies the conditions of Theorem~\ref{th:termSort} and completion is
still guaranteed to terminate. The result is the automaton $\aapprox'^*$ having 19
transitions. This time, intersection with $\A_{\TCS}$ gives:
 
{\footnotesize
\begin{alltt}
\States q4 q3 q2 q1 q0  \FinalStates q4  \Transitions a->q1  b->q3  nil->q0  
cons(q1,q0)->q2  cons(q1,q2)->q2  cons(q3,q2)->q4  cons(q3,q4)->q4
\end{alltt}  
}
 
\noindent
This automaton exactly recognizes lists of the form $[b,b, \ldots,a, a, \ldots]$ with at
least one $b$ and one $a$, as expected. Hopefully, refinement of equations
can be automatized in completion~\cite{BoichutBGL-ICFEM12} and can be used here,
see~\cite{GenetS-rep13} for examples. More examples can be found in the
\timbuk~3.1 source distribution.

\section{Conclusion and further research}
\label{sec:further}

In this paper we defined a criterion on the set of approximation equations to
guarantee termination of the tree automata completion. When dealing with, so called,
functional TRS this criterion is close to what is generally expected in static
analysis and abstract interpretation: a finite model for an infinite
set of data-terms. This work is a first step to use
completion for static analysis of functional programs. There remains some
interesting points to address.

\medskip
\noindent
{\em Dealing with higher-order functions}. Higher-order
functions can be encoded into first order TRS using a simple encoding 
borrowed from~\cite{Jones-Book87}: defined symbols become constants, constructor
symbols remain the same, and an additional {\em application} operator '\arobase'
of arity 2 is introduced. On all the examples
of~\cite{OngR-POPL11}, 
completion and this simple encoding produces exactly the same
results~\cite{GenetS-rep13}. 

\medskip
\noindent
{\em Dealing with evaluation strategies}. The technique proposed here, as well
as~\cite{OngR-POPL11}, over-approximates the set of results for all evaluation
strategies. As far as we know, no static analysis technique for functional
programs can take into account evaluation strategies. 
However, it is possible to restrict the completion algorithm to
recognize only innermost descendants~\cite{GenetS-rep13}, {\em i.e.}
call-by-value results. If the approximation is precise enough, any non
terminating program with call-by-value will have an empty set of
results. An open research direction is to use this to prove non termination
of functional programs by call-by-value strategy.

\medskip
\noindent {\em Dealing with built-in types}. Values manipulated by {\em real} functional
programs are not always terms or trees. They can be numerals or be terms
embedding numerals. In~\cite{GenetGLM-CIAA13}, it has been shown that completion
can compute over-approximations of reachable terms embedding built-in
terms. The structural part of the term is approximated using tree automata
and the built-in part is approximated using lattices and abstract
interpretation.


\medskip
\noindent
Besides, there remain some interesting theoretical points to solve. In
section~\ref{sec:termCrit}, we saw that having a finite $\TF/_{=_E}$ is not enough to guarantee the termination of
completion. This is due to the fact that the simplification algorithm does not
merge all states recognizing $E$-equivalent terms. Having a simplification
algorithm ensuring this property is not trivial. First,  the theory defined 
by $E$ has to be decidable. Second, even if $E$ is decidable, finding all the 
$E$-equivalent terms recognized by the tree automaton is an open problem.
Furthermore, proving that $\TF/_{=_E}$ is finite, is itself difficult. This
question is undecidable in general~\cite{TisonPrivate}, but can be answered for
some particular $E$. For instance, if $E$ can be oriented into
a TRS $\R$ which is terminating, confluent and such that $\lirr(\R)$ is finite
then $\TF/_{=_E}$ is finite~\cite{TisonPrivate}.

\medskip
\noindent
{\bf Acknowledgments} Many thanks to the referees for their detailed comments.

\bibliographystyle{plain}
{\small \bibliography{sabbrev,eureca,genet}}

\end{document}

%% file: rapinclude.tex
\definecolor{myviolet}{rgb}{.82, .0 ,.76}
\definecolor{myorange}{rgb}{1., .25 ,.0}
\definecolor{myblue}{rgb}{.01, .72 ,.625}
\definecolor{mygrey}{rgb}{.6, .6 ,.6}
\definecolor{myyellow}{rgb}{.75, .8 ,.0}
\definecolor{mygreen}{rgb}{.2, .8 ,.0}
\definecolor{mybrown}{rgb}{.68,.3,.1}
\definecolor{OliveGreen}{rgb}{.0, 0.36, 0.16}
\definecolor{mypink}{rgb}{.9, .0 ,.4}
\renewcommand{\equiv}{=}

\newcommand{\A}{\mathcal{A}}

\newcommand{\Q}{\mathcal{Q}}
\def \R {\mathcal{R}}

\renewcommand{\varepsilon}{\epsilon}
\newcommand{\epsifree}{\smash{\not}\varepsilon}
\newcommand{\Lang}{\mathcal{L}}
\newcommand{\Lange}{\mathcal{L}^{\epsifree}}

\newcommand{\desc}{\R^*}
\newcommand{\edesc}{\R_E^*}

\newcommand{\timbuk}{{\sf Timbuk}}

\newcommand{\F}{{\cal F}}
\newcommand{\myF}{{\cal K}}
\newcommand{\X}{{\cal X}}

\newcommand{\C}{{\cal C}}
\newcommand{\D}{{\cal D}}
\newcommand{\T}{{\cal T}}
\newcommand{\TF}{{\cal T(F)}}

\newcommand{\TFX}{{\cal T(F, X)}}
\newcommand{\TFQ}{{\cal T(F \cup Q)}}

\newcommand{\TC}{{\cal T(C)}}
\newcommand{\et}{\mbox{ and }}
\newcommand{\sth}{\mbox{ s.t. }}

\def\langtermesclos#1{\T(#1)}
\def\langtermesaux(#1,#2){\T(#1, #2)}
\def\sorts{\mathcal S}
\def\myTF{\langtermesclos\myF}
\def\myTFX{\langtermesaux(\myF,\X)}
\def\myTFs{\myTF^\sorts}
\def\ETF{E^{r}}
\def\Ec{E^c}

\def\Ecpy{E^c_{\myF}}
\def\Ecps{E^c_{{\myF,\sorts}}}
\def\Ecf{E^c_{\F}}
\def\Ecc{E^c_{\C}}
\def\Eccs{E^c_{\C,\sorts}}
\def\TFS{\TF^\sorts}
\def\TCS{\TC^\sorts}
\def\TFXS{\TFX^\sorts}

\DeclareMathOperator{\lirr}{I\textsc{rr}}
\newcommand{\airr}[1]{\A_{\lirr(#1)}}

\def\ERcpy{\overrightarrow{\Ecpy}}
\def\ERcps{\overrightarrow{\Ecps}}
\def\ERcf{\overrightarrow{\Ecf}}
\def\ERcc{\overrightarrow{\Ecc}}
\def\ERccs{\overrightarrow{\Eccs}}

\newcommand{\rw}{\rightarrow}

\newcommand{\simp}{\leadsto}
\newcommand{\Norm}{Norm}

\newcommand{\nr}{E}
\newcommand{\comp}{{\cal C}}
\def\simpleq#1#2{\mathcal S_{#1}\left(#2\right)}
\newcommand{\aapprox}{{\cal A}_{\R, \nr}}

\newcommand{\aaex}{{\cal A}_{\R}}

\def\build#1_#2^#3{\mathrel{
    \mathop{\kern 0pt#1}\limits_{#2}^{#3}}}

\def\CM#1{\build\hbox to 10mm {\rightarrowfill}_{}^{CM#1}}

\newcommand{\rws}[3]{\mathrel{{\build{\rightarrow}_{#1}^{#2}}\mskip-2mu_{#3}}}

\newcommand{\rwne}{\rightarrow^{\not\epsilon\: *}_{\simpleq E \A}}

\newcommand{\rwDne}{\rightarrow^{\not\epsilon}_{\Delta}}
\newcommand{\rwD}{\rightarrow_{\Delta}}
\newcommand{\rwAnestar}{\rightarrow^{\not\epsilon\: *}_{\A}}
\newcommand{\rwAznestar}{\rightarrow^{\not\epsilon\: *}_{\A_0}}

\newcommand{\rwDnestar}{\rightarrow^{\not\epsilon\: *}_{\Delta}}

\newcommand{\rwRE}{\rightarrow_{\R/E}}

\newcommand{\pos}{{\mathcal Pos}}

\newcommand{\var}{{\mathcal Var}}

\newcommand{\rwR}{\rws{}{}{\R}}
\newcommand{\rwA}{\rws{}{}{\A}}
\newfont{\amstoto}{msbm10}
\newcommand{\NN}{\mbox{\amstoto\char'116}}
\newcommand{\sep}{\; | \;}

\newcommand{\aut}{\langle \F, \Q, \Q_f, \Delta \rangle}

\lstdefinelanguage{yswhile}%
{morekeywords={if,else,while,then,endif,do,done, true, false, skip, Sequence, var, input, output},%
   sensitive=true,%
   morecomment=[n]{//}{},%
   morestring=[d]"%
  }[keywords,comments,strings]%

\def\timFont{\bf}
\newcommand{\Ops}{{\timFont Ops}}
\newcommand{\Vars}{{\timFont Vars}}
\newcommand{\TRS}{{\timFont TRS}}
\newcommand{\Automaton}{{\timFont Automaton}}

\newcommand{\States}{{\timFont States}}
\newcommand{\FinalStates}{{\timFont Final States}}
\newcommand{\Transitions}{{\timFont Transitions}}
\newcommand{\Equations}{{\timFont Equations}}
\newcommand{\Rules}{{\timFont Rules}}

\def\arobase{\symbol{64}}

\theoremstyle{plain}

\newtheorem*{lemma*}{Lemma}
\newtheorem*{theorem*}{Theorem}